\documentclass[journal,final,twocolumn,10pt,twoside]{IEEEtran}

\IEEEoverridecommandlockouts

\usepackage{cite}
\usepackage{amsmath,amsthm,amssymb,amsfonts}
\usepackage{algorithmic}
\usepackage{graphicx}
\usepackage{textcomp}
\usepackage{xcolor}
\usepackage{lipsum}
\usepackage{epstopdf}
\usepackage{booktabs}
\usepackage{amssymb}
\usepackage{multicol}
\usepackage{relsize}
\usepackage{epstopdf}
\usepackage{makecell}
\usepackage{array}
\usepackage{dsfont}
\usepackage{siunitx}
\usepackage{gensymb}
\usepackage[caption=false]{subfig}

\def\BibTeX{{\rm B\kern-.05em{\sc i\kern-.025em b}\kern-.08em
    T\kern-.1667em\lower.7ex\hbox{E}\kern-.125emX}}

\definecolor{morange}{rgb}{0.8,0.2,0}
\definecolor{mblue}{rgb}{0,0,1.0}
\definecolor{mred}{rgb}{0.9,0.1,0.1}
\definecolor{mgreen}{rgb}{0 0.8 0.3}

\newtheorem{theorem}{Theorem}
\newtheorem{lemma}{Lemma}

\DeclareMathOperator*{\argmax}{arg\,max}

\begin{document}

\title{A Concentration-Time Hybrid Modulation Scheme for Molecular Communications}

\author{Mustafa~Can~Gursoy,~\IEEEmembership{Student~Member,~IEEE,}
	Daewon~Seo,~\IEEEmembership{Member,~IEEE,}
	and~Urbashi~Mitra,~\IEEEmembership{Fellow,~IEEE} 
	\thanks{M. C. Gursoy and U. Mitra are with the Department of Electrical and Computer Engineering, University of Southern California, Los Angeles, CA 90089, USA (e-mails: \{mgursoy,~ubli\}@usc.edu).}
	\thanks{D. Seo is with the Department of Electrical and Computer Engineering,  University of Wisconsin-Madison, Madison, WI 53715, USA (e-mail: dseo24@wisc.edu).}
	\thanks{\color{black} This paper was presented in part at the IEEE International Conference on Communications (ICC 2020) \cite{our_ICC}. \color{black}This work has been funded in part by one or more of the following grants: ONR N00014-15-1-2550, NSF CCF-1817200, ARO W911NF1910269, Cisco Foundation 1980393, DOE DE-SC0021417, Swedish Research Council 2018-04359, NSF CCF-2008927, ONR 503400-78050.}
	\thanks{Digital Object Identifier: 10.1109/TMBMC.2021.3071772}
}

\markboth{TO APPEAR IN IEEE TRANSACTIONS ON MOLECULAR, BIOLOGICAL, AND MULTI-SCALE COMMUNICATIONS}
{TO APPEAR IN IEEE TRANSACTIONS ON MOLECULAR, BIOLOGICAL, AND MULTI-SCALE COMMUNICATIONS}

\maketitle

\begin{abstract}
Significant inter-symbol interference (ISI) challenges the achievement of reliable, high data-rate molecular communication via diffusion. In this paper, a hybrid modulation based on pulse position and concentration is proposed to mitigate ISI. By exploiting the time dimension, molecular concentration and position modulation (MCPM) increases the achievable data rate over conventional concentration and position-based modulations. In addition, unlike multi-molecule schemes, this hybrid scheme employs a single-molecule type and so simplifies transceiver implementations. In the paper, the optimal sequence detector of the proposed modulation is provided as well as a reduced complexity detector (two-stage, position-concentration detector, TPCD). A tractable cost function based on the TPCD detector is proposed and employed to optimize the design of the hybrid modulation scheme. In addition, the approximate probability of error for the MCPM-TPCD system is derived and is shown to be tight with respect to simulated performance. Numerically, MCPM shows improved performance over standard concentration and pulse position-based schemes in the low transmission power and high bit-rate regime. Furthermore, MCPM offers increased robustness against synchronization errors.
\end{abstract}


\section{Introduction}
\label{sec:introduction}

\par Molecular communications is a promising bio-inspired communication for establishing nano-networks \cite{molecular_1}. Among different ways of establishing molecular communication links, molecular communication via diffusion (MCD) has received particular attention due to its energy efficiency and bio-compatibility \cite{birkansurvey}. In an MCD system, emitted molecules rely solely on free diffusion after emission from the transmitter to arrive at the receiver. Since the molecules propagate randomly in the environment, the arrival times at the receiver are random variables \cite{AIGN_channel}. This physical phenomenon causes inter-symbol interference (ISI), which challenges reliable, high data-rate communication \cite{ISI_ref}.

\par Modulation design remains an open problem in MCD systems due to the unique features of the communication channel. Standard approaches include encoding information in the emission intensity (concentration shift keying, CSK, \cite{CSKMOSK}), emitted molecule type (molecule shift keying, MoSK, \cite{CSKMOSK,isomermosk}), or emission time (pulse position modulation, PPM, \cite{molcom_PPM}) of the molecular signal. 

\par \color{black} To combat ISI and increase data-rates, multiple molecule types have been employed to create orthogonal communication streams at the expense of more complex transceivers \cite{MCSK,D-MoSK}. Inhibitory molecules are employed in \cite{zebraCSK} to further reduce ISI, while in \cite{ISIBurcu,hybrid_2mol}, the molecule type is modified as a function of the past transmitted bits. Finally, multiple molecule types and PPM are the basis of the hybrid modulation in \cite{new_multimol_hybrid}.  While these schemes achieve their goals, they rely on synthesizing, storing, and counting multiple types of molecules, thus incurring significantly increased complexity over a single-molecule method. Motivated by this, we consider emission timing as a degree of freedom coupled with single molecule type signaling herein. 

\par As a timing-based modulation, PPM has received wide attention for radio-frequency based communications in ultrawideband \cite{PPM_UWB}, visible light communications \cite{PPM_VLC}, \textit{etc}. It has also been examined in the context of MCD \cite{molcom_PPM}, with maximum likelihood detection in the absence of ISI studied in \cite{murin1_opt} and higher order PPM for ISI mitigation investigated in \cite{bayramPPM}. 

\par In this paper, we encode information in emission concentration and time jointly. We observe that in an independent work \cite{JTAC}, hybrid concentration-time modulation is also considered. Therein, the capacity of \textbf{ISI-free} concentration-time channels is examined and shown to be larger than that of concentration or time alone. We underscore that herein, we do consider ISI and there are trade-offs to be made in the \textbf{design} of the hybrid modulation in order to achieve the best BER performance in such channels. In our preliminary work \cite{our_ICC}, we had showed that the hybrid modulation scheme achieves lower BER than binary CSK (BCSK) and PPM in MCD channels with severe ISI. This paper extends and completes \cite{our_ICC}. 

\color{black}

\par \color{black} Overall, the key contributions of this paper are as follows:
\begin{itemize}
    \item We propose a hybrid MCD modulation scheme that utilizes a single type of molecules. The proposed scheme merges $\mathcal{K}$-PPM and BCSK, which we call \textit{$\mathcal{K}$-ary molecular concentration and position modulation} ($\mathcal{K}$-MCPM).
    \item Considering the MCD channel faces ISI, we derive the maximum likelihood sequence detector (MLSD) for $\mathcal{K}$-MCPM.
    \item In addition to the MLSD that is computationally expensive, we propose a two-stage, position-concentration detector (TPCD) that reduces complexity. TPCD first detects the emission time, then performs a fixed threshold ($\gamma$) to resolve the concentration information.
    \item The binary concentrations are characterized by a parameter $\alpha \in (0.5,1)$. The parameter $\alpha$ is optimized under the assumption of the use of TPCD and through the derivation of a convex proxy for the error probability. We also present a method that estimates the optimal threshold parameter $\gamma$.
    \item Our numerical results suggest that the theoretically optimized $(\alpha,\gamma)$ pair yields close-to-optimal performance compared to the values found via exhaustive search for TPCD. 
    \item For a fixed $(\alpha,\gamma)$ pair, we derive the approximate error probability expression for a $\mathcal{K}$-MCPM scheme.
    \item Our numerical results show that MCPM outperforms BCSK and PPM, especially when the bit-rate is high and the transmission power is low. Furthermore, our results show that the MCPM scheme is more robust to synchronization offsets than PPM scheme of the same order. 
    \color{black}
\end{itemize}

\par Our prior work \cite{our_ICC} introduced the transmitter architecture of MCPM and the working principles of its TPCD detector. Herein, we complete these designs and analyses. In particular, we complete the modulation design by characterizing and solving the constellation point design problem for MCPM-TPCD. Furthermore, we derive the maximum likelihood sequence detector (MLSD) for MCPM, and introduce a low-complexity threshold selection method for TPCD. We also discuss the effects of temporal mis-synchronization on MCPM.  Full derivations and proofs are provided herein in contrast to \cite{our_ICC}.

\par The rest of the paper is organized as follows: Section \ref{sec:systemmodel} introduces the channel model. Section \ref{sec:proposedscheme} describes the proposed modulation scheme and discusses its key trade-off. Section \ref{sec:receiverdesign} introduces the optimal detector and a low complexity alternative which first detects position and then resolves the concentration. Section \ref{sec:BER} derives the error probability of the MCPM-TPCD system. Section \ref{sec:alphagamma} proposes theoretical methods to select the $(\alpha,\gamma)$ pair for an MCPM scheme. Section \ref{sec:results} presents numerical error probability results. Lastly, Section \ref{sec:conclusion} concludes the paper. Appendices A and B provide the proofs of the lemma and theorem presented in Section \ref{sec:alphagamma}.

\begin{table}[!h]
\centering
\caption{List of Abbreviations}
\label{tab:abbreviations}
\begin{tabular}{ll}
\hline
Abbreviation               & Full Name                   \\ \hline
BER & bit error ratio \\
BCSK& binary concentration shift keying \\
ES & exhaustive search \\
ISI & inter-symbol interference \\
LTI & linear time-invariant \\
MCD & molecular communication via diffusion \\
MCPM   & molecular concentration and position modulation  \\
ML & maximum likelihood \\
MLSD & maximum likelihood sequence detector \\
MoSK  & molecule shift keying \\
PPM & pulse-position modulation \\
TPCD          &  two stage, position-concentration detector \\ \hline                        
\end{tabular}
\end{table}
\color{black}

\section{System Model}
\label{sec:systemmodel}

\par The MCD system in this paper involves a single point transmitter and a single spherical absorbing receiver of radius $r_r$ in an unbounded $3$-D environment. The distance between the transmitter and the center of receiver is denoted by $r_0$. The transmitter and the receiver are assumed to have perfect synchronization unless stated otherwise. A visualization of the propagation environment is provided in Figure \ref{fig:topology}. In this case, assuming that carrier molecules have a diffusion coefficient $D$, the arrival probability density of a molecule $t$ seconds after its emission can be written as
\begin{equation}\label{eq:arrival_pdf}
f_{\textrm{hit}}(t) = \frac{r_{r}}{r_0} \frac{1}{\sqrt[]{4\pi Dt}} \frac{r_0 - r_r}{t} e^{- \frac{(r_{0}-r_{r})^2}{4Dt} },
\end{equation}
where $t \in (0,\infty)$.

\begin{figure}[!t]
	\centering
	\includegraphics[width=0.30\textwidth]{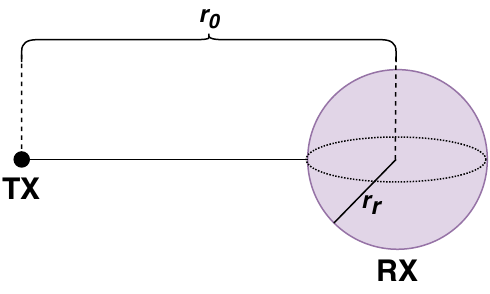} 
	\caption{The considered system model.}	
	\label{fig:topology}
\end{figure}

\par In a time slotted channel with sequential transmissions, the MCD channel is characterized by the channel coefficients, where the $n^{th}$ channel coefficient $h_n$ can be computed as
\begin{equation}
    \label{eq:FIRcoefs}
	h_n = \int_{(n-1)t_s}^{nt_s} f_{\textrm{hit}}(t) dt, \hspace{0.5cm} n = 1, 2, \dots, L.
\end{equation}
Here, $t_s$ denotes the time within the receiver's counting intervals (\textit{i.e.}, slots) and $L$ is the considered channel memory. In reality, the considered channel has infinite memory due to the heavy right tail of the arrival distribution \cite{3Dchar}. However, given that the arrival density function has small magnitude after a certain duration, we approximate the channel as having a finite duration, denoted by $L$. We select $L$ by considering a total time after which we neglect the arrivals ($t_{\textrm{total}}$), and computing $L$ using $t_{\textrm{total}}$ and $t_s$. Note that it is desirable to have a large $t_{\textrm{total}}$ to satisfactorily capture the right tail of the arrival density function in Equation \eqref{eq:arrival_pdf}.

\par The channel coefficients $h_n$ can be interpreted as a single molecule's probability of arrival at the receiver at the $n^{th}$ slot after its release. In a time-slotted MCD system where multiple molecules are emitted for each transmitted symbol, we employ the linear time-invariant (LTI)-Poisson channel model to characterize the number of arriving molecules at each slot \cite{LTI_Poisson}. According to the LTI-Poisson model, the arrival count in the $m^{th}$ time slot, denoted by $R_{m}$ in this paper, is distributed as
\begin{equation}
\label{eq:receivedPois}
R_{m} \sim \mathcal{P}  \bigg(\sum _{n=1}^{L} N_{m-n+1} h_{n} \bigg),
\end{equation}
where $\mathcal{P}(\cdot)$ denotes the Poisson distribution with argument as the rate parameter. In addition, $N_m$ denotes the number of molecules emitted by transmitter at the beginning of the $m^{th}$ time slot.

\section{Proposed Scheme}
\label{sec:proposedscheme}

\begin{table*}[!t]
	\centering
	\caption{Average emitted molecules per transmission and slot durations for BCSK, PPM, and MCPM}
	\begin{tabular}{llllllll}
		\hline
		\textbf{Modulation Scheme}  & \textbf{BCSK} & \textbf{$2$-PPM}   & \textbf{$4$-PPM}  & \textbf{$8$-PPM} & \textbf{$2$-MCPM} & \textbf{$4$-MCPM} & \textbf{$8$-MCPM} \\ \hline
		\begin{tabular}[c]{@{}l@{}} Transmitted bits \\ per unit symbol \end{tabular} 
		& 1     & 1    & 2     & 3     & 2 & 3 & 4   \\ \hline
		\begin{tabular}[c]{@{}l@{}} Sub-intervals \\ per symbol \end{tabular}   & 1  & 2   & 4   & 8   & 2 & 4 & 8  \\ \hline
		Sub-interval duration ($t_s$) & $t_b$   & $\frac{1}{2}t_b$  & $\frac{2}{4}t_b$ & $\frac{3}{8}t_b$  & $\frac{2}{2}t_b = t_b$ & $\frac{3}{4}t_b$ & $\frac{4}{8}t_b$ \\ \hline
		Bit duration  & $t_b$    & $t_b$   & $t_b$  & $t_b$ & $t_b$ & $t_b$ & $t_b$  \\ \hline
		\begin{tabular}[c]{@{}l@{}} Molecules per emission  \\ ($N$) \end{tabular}   & \begin{tabular}[c]{@{}l@{}}$2M$ for bit-1,\\ $0$ for bit-0 \end{tabular} 
		& $M$ & $2M$ & $3M$   & \begin{tabular}[c]{@{}l@{}l@{}l@{}} \textbf{BCSK bit-1} \\ $2M \times 2\alpha$ \\ \textbf{BCSK bit-0} \\ $2M \times 2(1-\alpha)$\end{tabular}  & \begin{tabular}[c]{@{}l@{}l@{}l@{}} \textbf{BCSK bit-1} \\ $3M \times 2\alpha$ \\ \textbf{BCSK bit-0} \\ $3M \times 2(1-\alpha)$\end{tabular} & \begin{tabular}[c]{@{}l@{}l@{}l@{}} \textbf{BCSK bit-1} \\ $4M \times 2\alpha$ \\ \textbf{BCSK bit-0} \\ $4M \times 2(1-\alpha)$\end{tabular} \\ \hline
		Molecules per bit (on average) & $M$  & $M$  & $M$   & $M$  & $M$ & $M$  & $M$ \\ \hline
	\end{tabular}
	\label{tab:paramtable}
\end{table*}

\subsection{Modulation Description}

\par Our proposed modulation scheme employs both concentration and the specific emission time of the molecules to convey information. Specifically, the proposed scheme combines the well-known PPM constellations with the conventional BCSK scheme and yields a two dimensional constellation diagram. Overall, the combination of $\mathcal{K}$-PPM and BCSK is referred to as the $\mathcal{K}$-ary molecular concentration-position modulation ($\mathcal{K}$-MCPM). 

\par At the transmitter, a serial-to-parallel conversion is done to group the bit stream into groups of length $(\log_2 \mathcal{K}) +  1 $ bits. As a convention, we consider the first $\log_2 \mathcal{K}$ to modulate the \textit{PPM component} of the modulation whereas the last bit determines the \textit{BCSK component}. In other words, the emission sub-interval of the molecular pulse is determined by the first $\log_2 \mathcal{K}$ bits, and the intensity of the said pulse is high or low depending on the BCSK bit being a ``$1$" or ``$0$". The transmission strategy for $4$-MCPM is presented in Figure \ref{fig:circled} for visualization purposes. 

\begin{figure}[!t]
	\centering
	\includegraphics[width=0.47\textwidth]{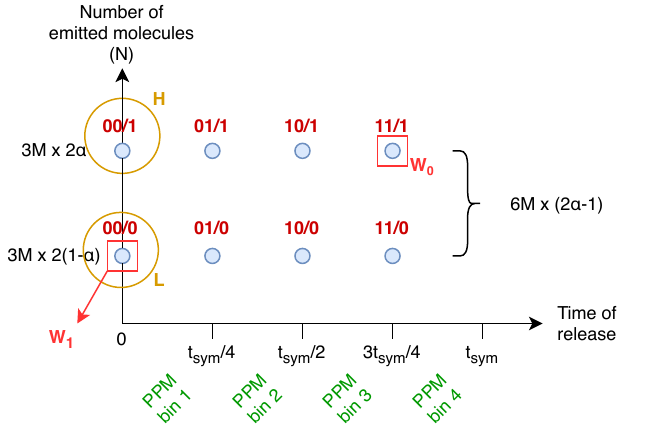} 
	\caption{The transmission strategy of $4$-MCPM. The last bit determines the emission intensity of the molecular pulse, whereas the first two bits govern the emission instant. Circled constellation points are of interest in $\alpha$ optimization.}	
	\label{fig:circled}
\end{figure}

\subsection{The Parameter $\alpha$}
\label{subsec:alpha_intro}

\par Throughout the paper, the transmission power is normalized on a per bit basis \cite{MCSK,bayramPPM}. In an MCD system, the energy consumption is related to the number of emitted molecules \cite{energy_model}. Hence, the transmission power normalization is done through imposing an average number of emitted molecules ($M$) constraint. The value of $M$ is a per-bit constraint such that each evaluated scheme emits, on average, $M$ molecules per bit. Additionally, we also employ a bit-rate normalization by imposing a constant bit duration, $t_b$. Thus, as more bits are used per symbol, the symbol duration ($t_{sym}$) is longer.

\par For a $\mathcal{K}$-MCPM scheme, the symbol duration can be written as $t_{sym} = (1 + \log_2 \mathcal{K}) t_b$, making each sub-slot duration $t_s = \frac{(1 + \log_2 \mathcal{K}) t_b}{\mathcal{K}}$. Furthermore, for $\mathcal{K}$-MCPM, the average number of molecules emitted per $\mathcal{K}$-MCPM symbol is equal to $E[N] = (1 + \log_2 \mathcal{K}) M$. Assuming all bits/symbols are equally likely, we have that
\begin{itemize}
    \item the MCPM symbols having the BCSK bit `$1$' are transmitted with $N = 2\alpha (1 + \log_2 \mathcal{K}) M$ molecules, and
    \item the MCPM symbols having the BCSK bit `$0$' are transmitted with $N = 2(1-\alpha) (1 + \log_2 \mathcal{K}) M$ molecules, 
\end{itemize}
where $\alpha \in (0.5,1)$. Table \ref{tab:paramtable} is provided to show the implications of these constraints in MCPM, traditional BCSK, and PPM.

\par The $\alpha$ parameter defines the two concentration levels for the BCSK part of the hybrid constellation, and thus determines the distances between constellation points. The value of $\alpha = 0.5$ results in no difference in concentration, yielding symbol ambiguities. On the other hand, if $\alpha$ is close to one, then the PPM signals for the BCSK symbols for the low concentration become hard to detect. Therefore, $\alpha$ is a design parameter to be optimized for an MCPM scheme. We formulate the $\alpha$ optimization problem in Section \ref{sec:alphagamma}. \color{black}

\section{Receiver Design}
\label{sec:receiverdesign}

\color{black}
\subsection{Optimal Detector}

\par Due to ISI, the maximum likelihood (ML) detectors for MCD modulations are in the form of ML sequence detectors (MLSD) \cite{molcom_receivers_akan}. We next present the MLSD for a $\mathcal{K}$-MCPM scheme herein. 

\par Let $\mathcal{S}$ denote the block length in terms of MCPM symbols. In addition, we define $\lambda_{m|\boldsymbol{s}}$ as the rate parameter of the arrival random variable at the $m^{th}$ interval, conditioned on the MCPM symbol sequence $\boldsymbol{s}$. Note that due to their definitions, $m \in \{ 1,\dots, \mathcal{S}\mathcal{K}\}$ and $\boldsymbol{s}$ is a vector of length $\mathcal{S}$. Recalling the LTI-Poisson model in \eqref{eq:receivedPois}, $\lambda_{m|\boldsymbol{s}}$ can be calculated as

\begin{equation} 
    \lambda_{m|\boldsymbol{s}} = \sum_{n=1}^{\mathcal{K} \cdot L_s} N_{m-n+1 | \boldsymbol{s}} h_{n} 
\end{equation}
where $N_{m | \boldsymbol{s}}$ denotes the number of emitted molecules by the transmitter at the $m^{th}$ time slot, conditioned on the candidate symbol sequence $\boldsymbol{s}$. Note that the rate parameter $\lambda_{m|\boldsymbol{s}}$ depends on the transmitted symbol sequence as it follows from \eqref{eq:receivedPois}. By defining the vector $\boldsymbol{r} = \begin{bmatrix} R_1 & \dots & R_{\mathcal{S} \cdot \mathcal{K}} \end{bmatrix}$, the MLSD for a $\mathcal{K}$-MCPM scheme can be expressed as
\begin{equation}
    \label{eq:MLSD}
    \begin{split}
        \hat{\boldsymbol{s}} &= \argmax_{\boldsymbol{s}} P(\boldsymbol{r} | \boldsymbol{s}) \\
        &= \argmax_{\boldsymbol{s}} \prod_{m=1}^{\mathcal{S} \cdot \mathcal{K}} \frac{\lambda_{m|\boldsymbol{s}}^{R_m} e^{-\lambda_{m|\boldsymbol{s}}}} {R_m!} \\
        &= \argmax_{\boldsymbol{s}} \sum_{m=1}^{\mathcal{S} \cdot \mathcal{K}} R_m \ln (\lambda_{m|\boldsymbol{s}}) - \lambda_{m|\boldsymbol{s}},
    \end{split}
\end{equation}
where $\hat{\boldsymbol{s}}$ denotes the vector of decoded symbols. For a block of length $\mathcal{S}$ and a memory of $L_s$ MCPM symbols (\textit{i.e.}, $L = \mathcal{K} \cdot L_s$), the $\mathcal{K}$-MCPM MLSD is of complexity $\mathcal{O}((2\mathcal{K})^{L_s}\mathcal{S})$ using a Viterbi decoder \cite{viterbialg}. The MLSD is of high complexity, but will serve as a benchmark to illustrate the performance complexity trade-off for our proposed decoder introduced in the sequel.
 
\subsection{A Reduced Complexity Detector}

\par Herein, we present an MCPM detector with low complexity, which we call the two-stage, position-concentration detector (TPCD). It employs two stages as its name implies: One for the PPM information and another for the BCSK. 

\par Recall that $s_k$ represents the $k^{th}$ MCPM symbol. At this point, we assume that $h_1 > \max(h_2,\dots,h_L)$, which implies that the first path is dominant\footnote{For MCD systems that yield practical error probabilities, this assumption is generally satisfied. However, it may not hold for extremely small symbol durations due to the behavior of \eqref{eq:arrival_pdf}.}. Given this assumption, the intended sub-interval is expected to have the largest arrival count among the $\mathcal{K}$ PPM bins. Denoting the emission slot of $s_k$ as $q_k$, the first stage of the detector performs 
\begin{equation}
    \label{eq:PPM_det}
    \hat{q}_k = \argmax_{j \in \{(k-1)\mathcal{K}+1,\dots, k\mathcal{K}\}} R_j.
\end{equation}

\par Since arrival counts are Poisson random variables, the arrival count random variable with the largest expected value has the highest mean-over-standard deviation ratio. The second stage of the MCPM detector performs fixed threshold detection on the largest arrival count to detect the BCSK bit. Denoting the decision threshold as $\gamma$, the rule can be written as
\begin{equation}
    \label{eq:CSK_det}
    R_{\hat{q}_k} \mathop{\gtrless}_{H_0}^{H_1} \gamma,
\end{equation}
where ${H_0}$ and ${H_1}$ correspond to the hypotheses that the $k^{th}$ MCPM symbol's BCSK bit being a `$0$' and a `$1$', respectively.

\color{black}
\subsection{Comparing the Detectors}
\label{subsec:detectoreval}

\par To compare the error performances of the two detection strategies, Figure \ref{fig:MLSD_TPCD_eval} is presented. Note that a fixed and small symbol memory $L_s = 3$ is selected for simulation, since the Viterbi decoding in the MLSD has exponential complexity in $L_s$.

\par As expected, the results of Figure \ref{fig:MLSD_TPCD_eval} suggests that TPCD incurs a performance loss comparative to the optimal detector. Furthermore, it is observable that the performance gap increases with $\mathcal{K}$. Note that for a $\mathcal{K}$-MCPM scheme, TPCD only considers the largest of the $\mathcal{K}$ obtained arrival counts, essentially disregarding the information coming from other $\mathcal{K}-1$ branches. On the other hand, the MLSD utilizes the arrival count information from all $\mathcal{K}$ slots.

\par Since the Viterbi decoder performs $(2\mathcal{K})^{L_s}$ likelihood computations per each $\mathcal{K}$-MCPM symbol (\textit{i.e.}, $\mathcal{O}( (2\mathcal{K})^{L_s}\mathcal{S})$ per one block), it is considerably more complex than TPCD which performs only two comparisons (\eqref{eq:PPM_det} and \eqref{eq:CSK_det}) per symbol (\textit{i.e.}, $\mathcal{O}(2\mathcal{S})$). This introduces a performance-complexity trade-off in terms of receiver design. In this paper, we will focus on the low-complexity option, TPCD, to demodulate MCPM symbols. Overall, the diagram of an MCD system that utilizes the MCPM scheme with TPCD is presented in Figure \ref{fig:overall_diagram}. 

\begin{figure}[!t]
	\centering
	\includegraphics[width=0.48\textwidth]{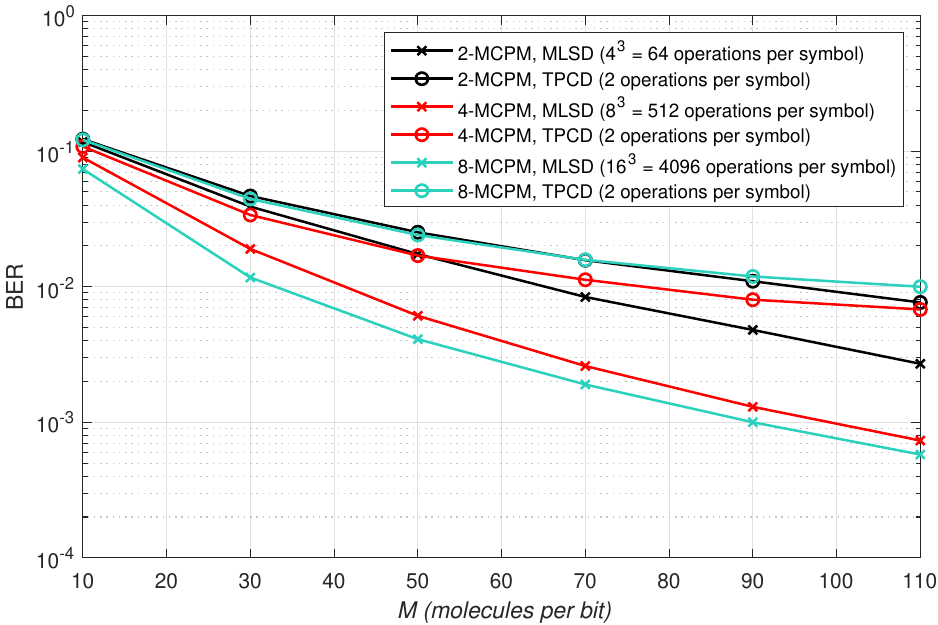}
	\caption{\color{black} BER vs. $M$ curves for MCPM using MLSD and TPCD. $t_b = \SI{0.30}{\second}$, $r_0 = \SI{10}{\micro\meter}$, $r_r = \SI{5}{\micro\meter}$, $D = \SI{79.4}{\micro\meter\squared\per\second}$, $L_s = 3$. $\alpha$ and $\gamma$ numerically optimized. \color{black}}
	\label{fig:MLSD_TPCD_eval}
\end{figure}

\begin{figure*}[!t]
	\centering
	\includegraphics[width=0.75\textwidth]{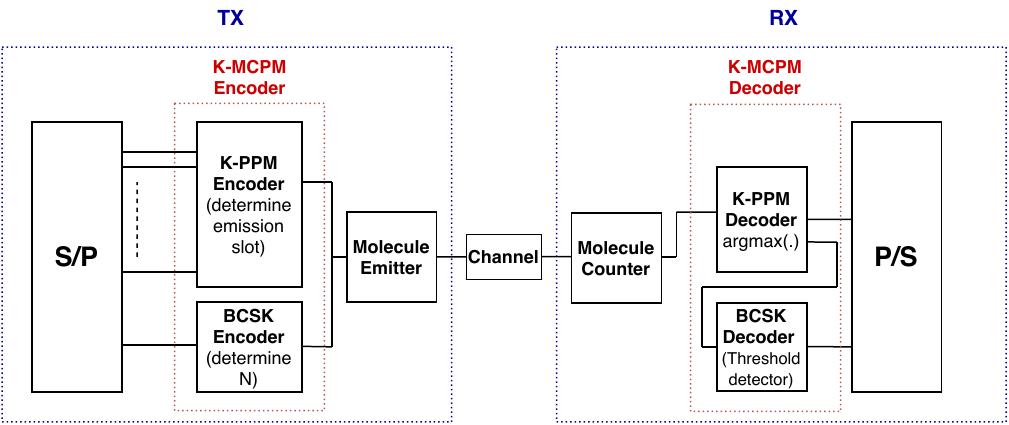} 
	\caption{Overall transmission and reception model of an MCD system utilizing $\mathcal{K}$-MCPM and the TPCD.}	
	\label{fig:overall_diagram}
\end{figure*}

\color{black}
\section{Error Analysis}
\label{sec:BER}

\par We have two design parameters which affect the performance of our overall system. To this end, we first provide an error analysis for the overall system given a fixed ($\alpha$,$\gamma$) pair. Given the ISI characteristics of the MCD channel \cite{molcom_receivers_akan,matched_filter,our_ICC}, the error probability can be found by averaging over conditional error probabilities. For a symbol memory of length $L_s$, this averaging is done over all possible symbol sequences of length $L_s$. Denoting the MCPM symbol sequence between the $(k-L_s+1)^{th}$ and $k^{th}$ transmissions as $s_{k-L_s+1:k}$, we write
\begin{equation}
\label{eq:overall}
P_e = \sum_{\forall s_{k-L_s+1:k}}^{} \Big(\frac{1}{2\mathcal{K}}\Big)^{L_s}  P_{e|s_{k-L_s+1:k}}.
\end{equation}
Denoting $d_H(\cdot,\cdot)$ as the Hamming distance operator, we proceed by obtaining the conditional error probability expression on the right-hand side as
\begin{equation}
\label{eq:conditional}
P_{e|s_{k-L_s+1:k}} = \sum_{n=1}^{2\mathcal{K}} \frac{d_H(\textbf{v}_{s_k},\textbf{v}_{n})}{1 + \log_2\mathcal{K} } P(\hat{s}_k = n | s_{k-L_s+1:k}).
\end{equation}
Here, $\textbf{v}_{(\cdot)}$ denotes the binary vector corresponding to the integer symbol in its argument, and $\hat{s}_k$ denotes the detected symbol. 

\par From \eqref{eq:conditional}, we seek to characterize $P(\hat{s}_k = n | s_{k-L_s+1:k})$. Denoting the conditional event $\{\hat{s}_k = n | s_{k-L_s+1:k}\}$ as $\hat{n}$, we note that the expression for $P(\hat{n})$ depends on the BCSK constellation of the intended bit. Let $b_{m} \in \{b_1,\dots,b_{\mathcal{K}}\}$ be the corresponding integer representation of the PPM bins of the MCPM symbol determined by $\textbf{v}_n$. Therefore, $R_{b_m}$ is a random variable conditioned on the $s_{k-L_s+1:k}$ sequence that describes the arrival count of $R_{(k-1)\mathcal{K}+m}$. $P(\hat{n})$ can be written as
\begin{equation}
\label{eq:exact_piecewise}
\begin{split}
&P(\hat{n}) = \\
&\begin{cases} 
P(R_{b_m} > \max (R_{b_m}^{'}) , R_{b_m} > \gamma) & \hspace{0.1cm} \textrm{if} \hspace{0.1cm} v_n[1 + \log_2\mathcal{K}] = 1, \\
P(R_{b_m} > \max (R_{b_m}^{'}) , R_{b_m} \leq \gamma) & \hspace{0.1cm} \textrm{if} \hspace{0.1cm} v_n[1 + \log_2\mathcal{K}] = 0, \\
\end{cases}
\end{split}
\end{equation}
where $R_{b_m}^{'}$ denotes the set of arrival counts corresponding to each bin other than $b_m$. \color{black} Denoting $Y = \max (R_{b_m}^{'})$, we can find $P(\hat{n})$ where $v_n[1 + \log_2\mathcal{K}] = 0$ as
\begin{equation}\label{eq:Pn_0}
\begin{split}
P(R_{b_m} > \max (R_{b_m}^{'}) , R_{b_m} & \leq \gamma) = P(Y < R_{b_m} \leq \gamma) \\
=& \int_{-\infty}^{\gamma} \bigg[ \int_{y}^{\gamma} f_{R_{b_m}}(r) dr \bigg]   f_Y(y) dy \\
=& \int_{-\infty}^{\gamma} \bigg[ \int_{-\infty}^{r}  f_Y(y) dy \bigg] f_{R_{b_m}}(r) dr \\
=& \int_{-\infty}^{\gamma} F_Y(r) f_{R_{b_m}}(r) dr.
\end{split}
\end{equation}
\color{black}
Using the Gaussian approximation on the Poisson arrival counts \cite{arrivalmodel}, the CDF of $Y$ can be found as
\begin{equation} 
\begin{split}\label{eq:F_Y}
F_Y(r) =& P\big(\max (R_{b_m}^{'}) \leq r \big) \\
=& \prod_{\substack{\tau = 1 \\ \tau \neq b_m}}^{\mathcal{K}} P(R_{b_\tau} \leq r) \\
=& \prod_{\substack{\tau = 1 \\ \tau \neq b_m}}^{\mathcal{K}} \bigg[ 1- Q\bigg(\frac{r-\mu_{R_{b_\tau}}}{\sigma_{R_{b_\tau}}}\bigg)\bigg],
\end{split}
\end{equation}
where $Q(\cdot)$ is the complementary cumulative distribution function of the standard Gaussian. Here, $\sigma_{R_{b_\tau}} = \sqrt{\mu_{R_{b_\tau}}}$ as the arrival counts are Poisson random variables. In addition, each mean (rate) parameter conditioned on $s_{k-L_s+1:k}$ can be found as $\mu_{R_{b_\tau}} =  \sum_{n=1}^{\mathcal{K} \cdot L_s -1} N_{(k-1)\mathcal{K}+\tau-n+1} h_{n} $, due to the LTI-Poisson nature of the channel as described in \eqref{eq:receivedPois}. We also note that the non-zero $N$ values (\textit{i.e.}, emitted number of molecules at transmission instants) are either $B \alpha$ or $B (1-\alpha)$ depending on the concentration constellation, where $B = 2M (1 + \log_2\mathcal{K})$. Lastly, Similar to \eqref{eq:Pn_0}, $P(\hat{n})$ for the case where $v_n[1 + \log_2\mathcal{K}] = 1$ can be expressed as
\begin{equation}\label{eq:Pn_1}
\begin{split}
P(R_{b_m} > \max (R_{b_m}^{'}) , R_{b_m} > \gamma) = \int_{\gamma}^{\infty} F_Y(r) f_{R_{b_m}}(r) dr,
\end{split}
\end{equation}
completing our derivation. \color{black}

\begin{figure}[!t]
	\centering
	\includegraphics[width=0.48\textwidth]{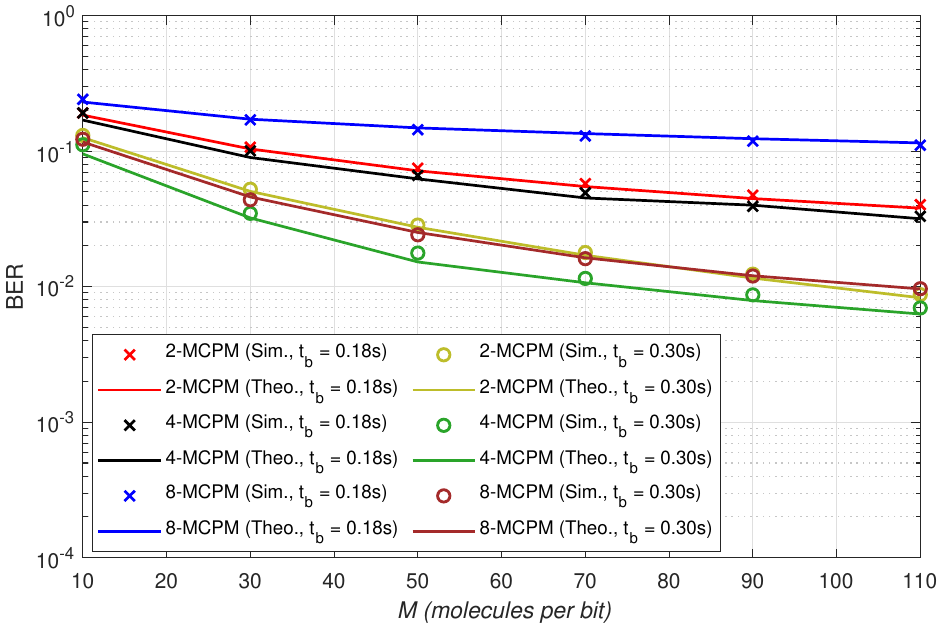} 
	\caption{BER vs. $M$ curves. $t_b = 0.18$ and $0.30$s, $r_0 = 10\mu$m, $r_r = 5 \mu$m, $D = 79.4 \frac{\mu m^2}{s}$, and $t_{\textrm{total}} = 12t_b$. ($\alpha$,$\gamma$) pairs optimized through exhaustive search.}	
	\label{fig:theo_L12}
\end{figure}

\begin{figure}[!t]
	\centering
	\includegraphics[width=0.48\textwidth]{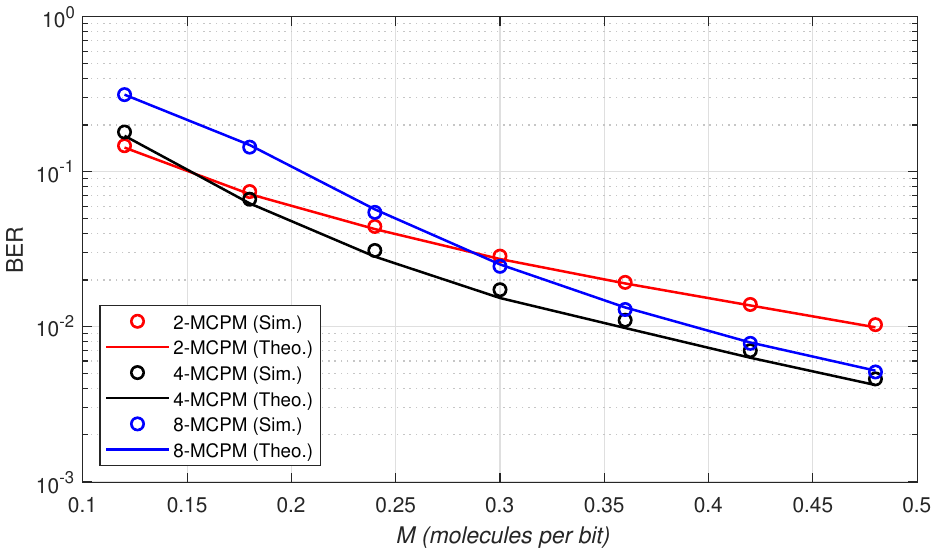} 
	\caption{BER vs. $t_b$ curves. $M = 50$ molecules, $r_0 = 10\mu$m, $r_r = 5 \mu$m, $D = 79.4 \frac{\mu m^2}{s}$, and $t_{\textrm{total}} = 12t_b$. ($\alpha$,$\gamma$) pairs optimized through exhaustive search.}
	\label{fig:theo_L12_tb}
\end{figure}

\par Figure \ref{fig:theo_L12} provides the probability of error versus $M$ for different MCPM orders to demonstrate the accuracy of Equations \eqref{eq:overall}-\eqref{eq:F_Y}. Furthermore, Figure \ref{fig:theo_L12_tb} presents BER versus $t_b$ to evaluate the accuracy of the derived BER under different ISI magnitudes. While, our probability of error derivation invokes approximations, Figures \ref{fig:theo_L12} and \ref{fig:theo_L12_tb} confirm that the closed form expressions provide good approximations to the true probability of error.

\section{The Optimization of $\alpha$ and $\gamma$}
\label{sec:alphagamma}

\par As noted in Section \ref{sec:proposedscheme}, for a fixed $M$, $\alpha$ is a parameter that poses a trade-off between the detection accuracies of the position and concentration constellations. For an MCPM system that uses TPCD at the receiver, we see that the ($\alpha$,$\gamma$) pair needs to be optimized with the goal of minimizing the probability of error. In this section, we address this optimization problem.

\par Directly optimizing the probability of error expression derived in Equations \eqref{eq:overall}-\eqref{eq:Pn_1} is computationally expensive as one would need to consider $(2\mathcal{K})^{L_s}$ conditional error probabilities for each evaluated ($\alpha,\gamma$) pair. To this end, we determine methods informed by the true probability of error and take a two-step approach: We first estimate $\alpha$ and then use the estimated $\alpha$ to optimize $\gamma$.

\subsection{Selecting $\alpha$}
\label{subsec:alpha}

\par Herein, inspired by the properties of an MCPM scheme and the nature of the MCD channel, we propose a low-complexity sub-optimal cost function that is provably convex under reasonable conditions. Note that $\alpha$ determines the distances between different MCPM constellation points. For tractability, when optimizing $\alpha$, we use a hypothetical no-ISI scenario (even though the evaluated channel has ISI) instead of considering the ISI between consecutive MCPM symbols. \color{black} Our results will show that this approximation is not detrimental to overall performance.

\par The no-ISI assumption corresponds to the case where the channel is cleared after each MCPM symbol. Note that the ISI is still present within the temporal bins of each MCPM symbol. In this hypothetical scenario, among each column of the 2-D constellation diagram, the left-most one is the most likely to be detected erroneously. Therefore, we focus our design into the left-most constellation points (\textit{i.e.}, the circled points ``H" and ``L" in Figure \ref{fig:circled}). 

\par The error probabilities associated with these points can be written as 
\begin{equation}
\begin{split}
P_{e|H} &= P(R_1 < R_2 \cup \dots \cup R_1 < R_{\mathcal{K}} \cup R_1 < \gamma | \textrm{H sent})\\
P_{e|L} &= P(R_1 < R_2 \cup \dots \cup R_1 < R_{\mathcal{K}} \cup R_1 > \gamma | \textrm{L sent}).
\end{split}
\label{eq:Peilk}
\end{equation}
We use the union bound on the expression $P_{e|H} + P_{e|L}$ to determine our cost function as
\begin{equation}
	\begin{split}
		C &= P(R_1 < \gamma | \textrm{H sent}) + P(R_1 > \gamma | \textrm{L sent})\\
		& + \sum_{i=2}^{\mathcal{K}} P(R_1 < R_i | \textrm{H sent}) + P(R_1 < R_i | \textrm{L sent}).
	\end{split}
	\label{eq:cost_function_proto}
\end{equation}
Using the Gaussian approximation on the arrival counts \cite{arrivalmodel}, we obtain
\begin{equation}
\begin{split}
&C = Q\bigg( \frac{B\alpha_{U} h_1 - \gamma_U}{ \sqrt{B \alpha_{U} h_1}} \bigg) + Q\bigg( \frac{\gamma_{U} - B(1-\alpha_{U}) h_1}{ \sqrt{B (1-\alpha) h_1}} \bigg) + \\
&\sum_{i=2}^{\mathcal{K}} Q\bigg( \frac{B \alpha_{U} (h_1 - h_i)}{ \sqrt{B \alpha_{U} (h_1 +h_i)}} \bigg) + Q\bigg( \frac{B (1-\alpha_{U}) (h_1 - h_i)}{ \sqrt{B (1-\alpha_{U}) (h_1 +h_i)}} \bigg),
\end{split}
\label{eq:cost_function_Q}
\end{equation}
where the ($\alpha_U, \gamma_U$) pair represents the $\alpha$ and $\gamma$ values for the hypothetical no-ISI scenario. 

\par Minimizing $C$ requires the optimization of $\alpha_{U}$ and $\gamma_{U}$ jointly. However, by deriving the optimal $\gamma_{U}$ value in terms of $\alpha_U$, we can reduce the dimension of the numerical search. We now show the convexity of $C$ in $\alpha_U$ under the following set of conditions:
\begin{equation}
	\begin{split}
		&0.5 < \alpha_U < 1 \\
		&B(1-\alpha_U)h_1 < \gamma_U < B \alpha_U h_1 \\
		&h_1 > \max(h_2,\dots,h_\mathcal{K}) > 0 \\
		&B \alpha h_1 - B(1-\alpha)h_1 > 3
	\end{split}
	\label{eq:assumptions}
\end{equation}
Note that the first two conditions define the valid intervals of the parameters $\alpha_U$ and $\gamma_U$. Here, the interval of $\alpha_U$ follows from the definition of the parameter in Subsection \ref{subsec:alpha_intro}, and $\gamma_U$ needs to lie within the interval $(B(1-\alpha_U)h_1,B\alpha_U h_1)$ to be a meaningful threshold. The third condition assumes that the symbol duration is defined such that the first channel coefficient is the largest one. The last condition lower bounds the distance between the expected arrival counts of adjacent concentration constellations and is typically satisfied in practical MCD links, including all data points generated in the paper.

\par We first start by finding the optimal $\gamma_U$ for each $\alpha_U$ that minimizes $C$ by providing the following lemma:
\begin{lemma}[Convexity in $\gamma_U$]
	For all $\alpha_U \in (0.5,1)$, the cost function $C$ is convex in $\gamma_U$, given the validity condition $B(1-\alpha_U)h_1 < \gamma_U < B \alpha_U h_1$ is held.
	\label{lemma:gamma_convex}
\end{lemma}
\begin{proof}
	The proof is provided in Appendix \ref{ap:gamma_cvx}.
\end{proof}
Given Lemma \ref{lemma:gamma_convex}, the optimal $\gamma_U$ can be found from the vanishing point of $\frac{\partial C}{\partial \gamma_{U}}$, written as
\begin{equation}
	\gamma^{*}_U(\alpha_{U}) = \sqrt{ \frac{ Bh_1 + \ln(\frac{1-\alpha_{U}}{\alpha_{U}}) - 2B \alpha_{U} h_1 }{ \frac{1}{B \alpha_{U} (1-\alpha_{U})h_1} - \frac{2}{B(1-\alpha)h_1}  }  }.
	\label{eq:optimal_gamma}
\end{equation}
Given that we now have an expression of the optimal $\gamma_{U}$ for each $\alpha_U$, we can now show that $C$ is convex in $\alpha_{U}$ when $\gamma_U = \gamma^{*}_U(\alpha_{U})$.

\begin{theorem}[Convexity in $\alpha_U$]
	Given that the conditions in \eqref{eq:assumptions} is met, $C$ is convex in $\alpha_U$ when $\gamma_U = \gamma^{*}_U(\alpha_{U})$.
	\label{theorem:alpha}
\end{theorem}
\begin{proof}[Proof]
    The proof of Theorem \ref{theorem:alpha} is given in Appendix \ref{ap:alpha_cvx}.
\end{proof}

\par At this point, for the no-ISI scenario, we have found $\gamma^{*}_U$ in closed form as a function of $\alpha_U$, and shown the union bound cost function $C$ to be convex in $\alpha_U$ when operating at $\gamma_{U} = \gamma^{*}_U(\alpha_{U})$. Therefore, $C$ can be minimized using simple $1$-dimensional numerical search algorithms, and $\alpha^*$ can be found. Note that performing this operation only once (before the data transmission starts) is sufficient. It is also noteworthy that $\gamma_U$ is simply a dummy variable in this derivation. The obtained $\alpha^*$ is employed to determine $\gamma$ in the following.

\subsection{Selecting $\gamma$}
\label{subsec:gamma}

\par In Subsection \ref{subsec:alpha}, a hypothetical no-ISI scenario is considered to optimize $\alpha$. However, the same approach cannot be taken for $\gamma$, since completely ignoring earlier symbol transmissions causes the obtained $\gamma$ to be considerably smaller than the actual optimal value. Motivated by this shortcoming, we propose a low complexity, sub-optimal method for estimating the $\gamma$ value in the actual, with-ISI scenario. The proposed approach works as follows:
\begin{enumerate}
	\item The \textit{worst-case symbol sequences} for each point on the constellation diagram in terms of ISI are considered. 
	\begin{itemize}
		\item For the MCPM constellation points where the concentration bit is a $0$ (\textit{i.e.}, the bottom row), the worst-case sequence is the one causing the highest ISI. The highest ISI is generated when all the past symbols are transmitted at the $\mathcal{K}^{th}$ sub-slot and with the high concentration. For a symbol memory of $L_s$ and using $4$-MCPM, this sequence corresponds to a ($L_s - 1$) symbols-long repeated transmission of $W_0$ in Figure \ref{fig:circled}. Note that this is similar to the ``...-1-1-1-1-1-0" case in pure BCSK.
		\item Similarly, for the MCPM constellation points where the concentration bit is a $1$ (top row), the worst-case sequence is the one causing the lowest ISI, since ISI would \textit{help} the correct detection of the concentration bit otherwise \cite{worst_seqs}. For a symbol memory consideration of $L_s$ and using $4$-MCPM, this sequence corresponds to a ($L_s - 1$) symbols-long repeated transmission of $W_1$ in Figure \ref{fig:circled}. Note that this is similar to the ``...-0-0-0-0-0-1" case in pure BCSK.
	\end{itemize}
	Let $\mu^{w}_{i,j}$ be the conditional arrival mean of the constellation point located at the $i^{th}$ PPM bin and has the BCSK bit $j$, under the corresponding worst-case sequence considerations. Here, $i \in \{1,\dots,\mathcal{K}\}$ and $j \in \{0,1\}$. These conditional arrival means can be written as
	\begin{equation} \label{eq:worstcasemean}
	\hspace{-0.55cm} \mu^{w}_{i,j} = 
	\begin{cases} 
	B \alpha^* h_1 + \sum_{c = 1}^{L_s-1} B (1-\alpha^*)h_{c\mathcal{K}+i}  & j = 1, \\
	B (1-\alpha^*) h_1 + \sum_{c = 1}^{L_s-1} B \alpha^* h_{(c-1)\mathcal{K}+i+1}  & j = 0. \\
	\end{cases}
    \end{equation}
	
	\item After finding the conditional statistics, the Gaussian approximation of the Poisson distribution is again employed, in order to find an individual threshold between the upper and lower-row constellations of the $i^{th}$ PPM bin, which we call $\gamma^{w}_i$. Specifically, $\gamma^{w}_i$ locates at the point where the upper ($\sim \mathcal{N}(\mu^{w}_{i,1},\mu^{w}_{i,1})$) and lower ($\sim \mathcal{N}(\mu^{w}_{i,0},\mu^{w}_{i,0})$) constellations' densities intersect. Thus, $\gamma^{w}_i$ is the solution to the equation 
    \begin{equation}
    \label{eq:gaus_cross}
    \frac{1}{\sqrt{2 \pi \mu^{w}_{i,1}}} e^{\frac{(\gamma^{w}_i - \mu^{w}_{i,1})^2}{2\mu^{w}_{i,1}}} = \frac{1}{\sqrt{2 \pi \mu^{w}_{i,0}}} e^{\frac{(\gamma^{w}_i - \mu^{w}_{i,0})^2}{2\mu^{w}_{i,0}}}.
    \end{equation}
    
    \item Lastly, the obtained candidate $\gamma^{w}_i$ values are averaged to find the estimated $\gamma$ as
\begin{equation}
	\label{eq:gamma_averaging}
	\gamma^{*} = \Biggl\lfloor \frac{1}{\mathcal{K}} \sum_{i = 1}^{\mathcal{K}} \gamma^{w}_i \Biggr\rfloor + \frac{1}{2}.
\end{equation}
The floor function operation in \eqref{eq:gamma_averaging} is done since in the actual case, the arrival counts are discrete and Poisson distributed random variables.	
\end{enumerate}

\section{Numerical Results}
\label{sec:results}

\par Herein, numerical BER results of the proposed scheme are presented under different channel conditions and detectors through computer simulations, on the LTI-Poisson channel model described in Section \ref{sec:systemmodel},\footnote{Note that in both Equation \eqref{eq:PPM_det} and $\alpha$ optimization, we had assumed $h_1 > \max(h_2,\dots,h_L)$. Though it is generally satisfied, this assumption is not imposed on any evaluated system in this paper. The channel coefficients are generated using Equations \eqref{eq:arrival_pdf}-\eqref{eq:FIRcoefs}, according to the system parameters\color{black}.} and using the parameters as presented in Table \ref{tab:simtable}.\footnote{The chosen value for $D$ is reported to be a conservative value for the diffusion coefficient of insulin in water at 20$\degree$C \cite{insulin_D}.} Since MCPM is a combination of BCSK and $\mathcal{K}$-ary PPM, the proposed scheme is compared to these modulation schemes in order to present the gain it has over its building blocks. In addition, unless specified, all $(\alpha,\gamma)$ pairs of the MCPM schemes in the section are obtained via exhaustive search and using TPCD, and thus presumed to be the optimal values. BCSK is demodulated using a numerically optimized threshold detector similar to Equation \eqref{eq:CSK_det}, and PPM schemes are demodulated using the maximum count decoder similar to Equation \eqref{eq:PPM_det}. All figures herein employ the normalizations presented in Table \ref{tab:paramtable}.

\begin{table}[!t]
\centering
\caption{Considered system and channel parameters. Default values shown in bold.}
\label{tab:simtable}
\begin{tabular}{ll}
\hline
Parameter                 & Value                             \\ \hline
$M$ (molecules)         & $10, 30, \textbf{50}, 70, 90, 110$  \\
$t_b$ $(s)$              & $0.12, 0.18, 0.24, \textbf{0.30}, 0.36, 0.42, 0.48$    \\
$r_0$ $(\SI{}{\micro\meter})$    & $10$                              \\
$r_r$ $(\SI{}{\micro\meter})$         & $5$    \\
$D$ $(\SI{}{\micro\meter\squared\per\second})$  & $79.4$   \\ \hline                        
\end{tabular}
\end{table}

\subsection{Performance Evaluation of the $\alpha$ and $\gamma$ Optimizations}
\label{subsec:alphagamma_eval}

\par Both $\alpha$ and $\gamma$ selection methods are sub-optimal procedures due to their simplifying considerations. Thus, this subsection provides numerical results regarding the accuracy of the proposed methods to select the $(\alpha,\gamma)$ pair. Figure \ref{fig:alphagammatable_tb300ms} presents the theoretical and optimal $(\alpha,\gamma)$ pairs for a set of channel parameters. Throughout the subsection, the theoretical curves correspond to calculating both $\alpha$ and $\gamma$ using methods described in Section \ref{sec:alphagamma}, whereas the curves using true optimal values use ($\alpha$,$\gamma$) pairs are calculated using exhaustive search (ES).

\begin{figure}[!t]
	\centering
	\includegraphics[width=0.48\textwidth]{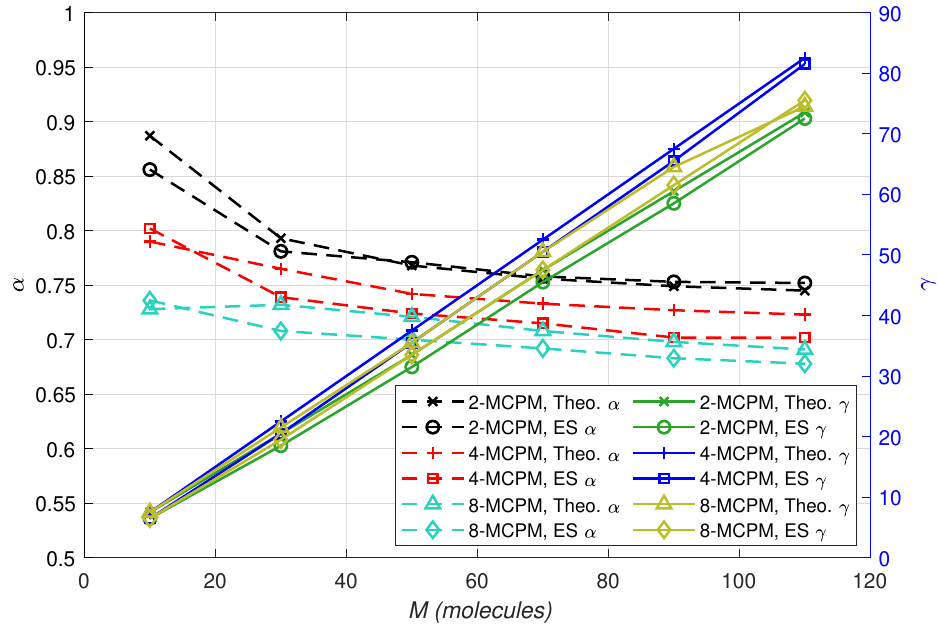} 
	\caption{ $\alpha$ and $\gamma$ vs. $M$ for MCPM, using theoretical and simulated values. $t_b = \SI{0.30}{\second}$, $r_0 = \SI{10}{\micro\meter}$, $r_r = \SI{5}{\micro\meter}$, $D = \SI{79.4}{\micro\meter\squared\per\second}$, and $t_{\textrm{total}} = 48t_b$.} 
	\label{fig:alphagammatable_tb300ms}
\end{figure}

\par Overall, the results of Figure \ref{fig:alphagammatable_tb300ms} suggest that the results of our optimizations follow the trends of the true $(\alpha,\gamma)$ pairs. That being said, the effect of the slight discrepancies need to be addressed to assess the sensitivity of the error performance with respect to $\alpha$ and $\gamma$. In order to evaluate the error performance achieved by the optimization approaches presented in Subsections \ref{subsec:alpha} and \ref{subsec:gamma}, two sets of results are presented. Firstly, Figure \ref{fig:M_theosim_tb300} is presented to evaluate the performance with respect to $M$, using the $(\alpha,\gamma)$ pairs obtained for Figure \ref{fig:alphagammatable_tb300ms}. In addition, Figure \ref{fig:tb_theosim_M50} is presented to evaluate the effect of $t_b$ on the performance. 

\begin{figure}[!t]
	\centering
	\includegraphics[width=0.48\textwidth]{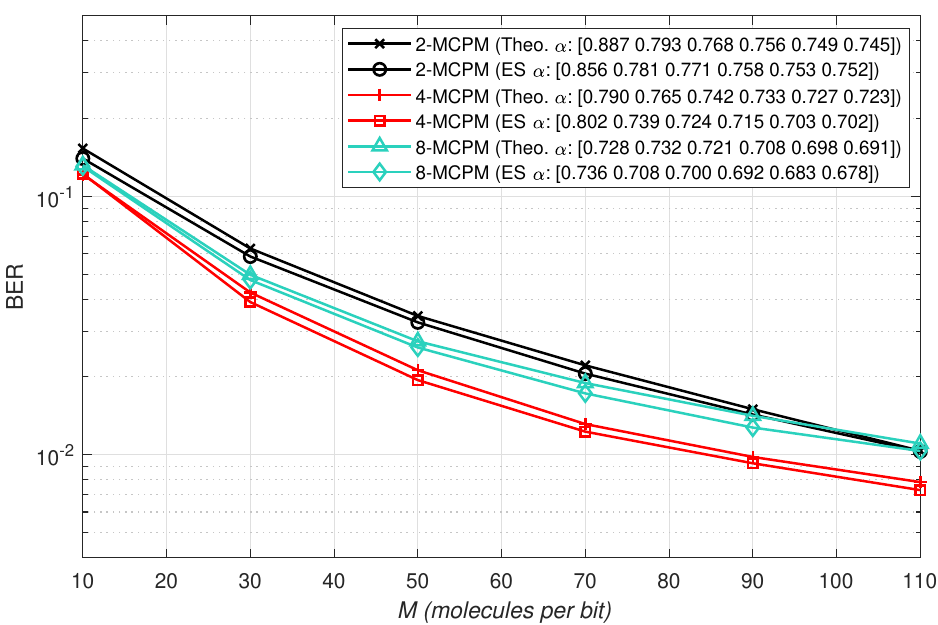} 
	\caption{ BER vs. $M$ curves for MCPM. $t_b = \SI{0.30}{\second}$, $r_0 = \SI{10}{\micro\meter}$, $r_r = \SI{5}{\micro\meter}$, $D = \SI{79.4}{\micro\meter\squared\per\second}$, and $t_{\textrm{total}} = 48t_b$.} 
	\label{fig:M_theosim_tb300}
\end{figure}

\begin{figure}[!t]
	\centering
	\includegraphics[width=0.48\textwidth]{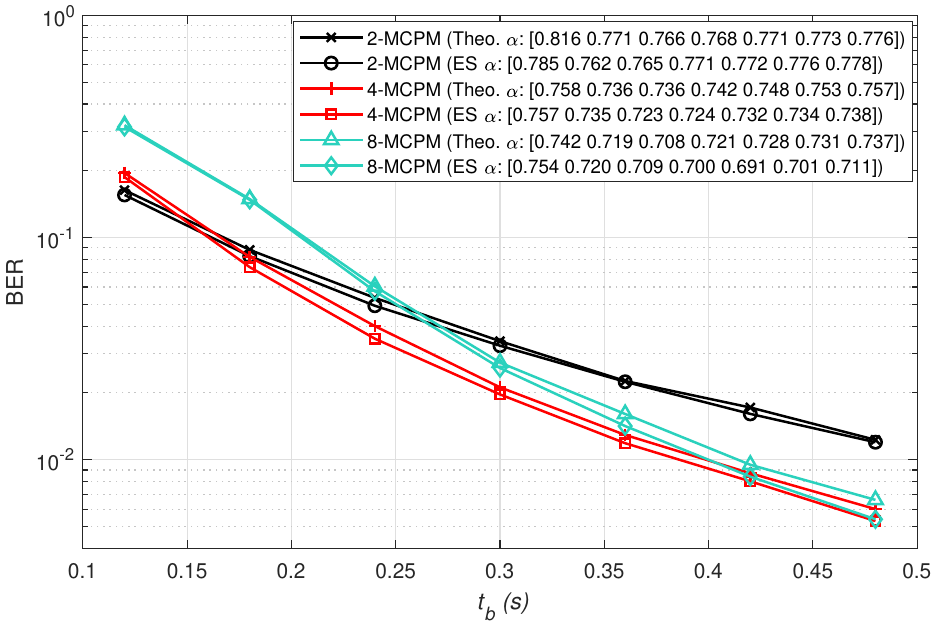} 
	\caption{ BER vs. $t_b$ curves for MCPM. $M = 50$ molecules, $r_0 = \SI{10}{\micro\meter}$, $r_r = \SI{5}{\micro\meter}$, $D = \SI{79.4}{\micro\meter\squared\per\second}$, and $t_{\textrm{total}} = 48t_b$.} 	
	\label{fig:tb_theosim_M50}
\end{figure}

\par The numerical results of Figures \ref{fig:M_theosim_tb300} and \ref{fig:tb_theosim_M50} suggest that the BER obtained by employing the $(\alpha,\gamma)$ pairs found via our optimization strategies closely approximate the BER obtained through numerical exhaustive searched $(\alpha,\gamma)$ pairs. Note that Figure \ref{fig:tb_theosim_M50} shows that our method is especially effective when the bit duration $t_b$ is smaller, which is desirable as the high bit-rate regime is indeed our operation regime of interest. That being said, a slight performance loss is incurred at larger $t_b$ values (lower data-rate). 

\par Our empirical observations suggest that the accuracy in estimating the optimal $\alpha$ does not considerably increase in the absolute error sense as $t_b$ increases. However, we observed that the sensitivity of the error performance to $\alpha$ increases with $t_b$, which partially causes the slight performance \color{black} loss incurred by the sub-optimal methods. One interesting finding is that even though a larger $t_b$ implies the accuracy of the no-ISI assumption to improve, the absolute error between the optimal and estimated $\alpha$ does not monotonically decrease with $t_b$. This behavior is due to the sub-optimality introduced in the union bound in Equation \eqref{eq:cost_function_proto}, as the bounding results in counting some decision regions multiple times \cite{molisch_book}. 

\par In addition to the aforementioned discussion, a portion of the performance gap can be explained due to the approximations done in the $\gamma$ optimization. As $t_b$ increases, the incurred ISI decreases for all symbol sequences, including the worst-case sequences for the MCPM scheme. This, in turn, decreases the ISI contribution of conditional arrival statistics. In fact, this decrease leads some candidate $\gamma^{w}_i$ values to undershoot, causing the obtained $\gamma^{*}$ from Equation \eqref{eq:gamma_averaging} to be smaller than the actual value. Overall, emphasizing on the goals of establishing high data-rate communication (small $t_b$) with low power consumption (small $M$) and low computational complexity, we believe that despite their shortcomings in the large $t_b$ regime, the proposed sub-optimal methods for determining $\alpha$ and $\gamma$ still have utility in the regime of interest.

\subsection{Error Performance and Transmission Power}
\label{subsec:Msweep}

\par In a nano-scale MCD link, energy consumption is an important design criterion. Acknowledging the presence of the relation between the consumed energy and the number of emitted molecules \cite{energy_model}, this subsection presents the error performance of MCPM with respect to $M$. Here, Figure \ref{fig:M_sweep_tb180} considers a high ISI scenario given the channel parameters, whereas Figure \ref{fig:M_sweep_tb300} has a more benign channel. 

\par The results of Figure \ref{fig:M_sweep_tb180} suggest that under high ISI, the proposed scheme outperforms its concentration and position counterparts. The main reason for this desirable gain is the ability of MCPM to encode more bits within its constellations, effectively mitigating ISI by having a longer symbol duration while still satisfying the same bit-rate constraint ($\frac{1}{t_b}$). One interesting observation is that when using the TPCD, although increasing $\mathcal{K}$ corresponds to having a sparser transmission in the temporal axis, the error performance of $\mathcal{K}$-MCPM does not monotonically improve with $\mathcal{K}$. We note this phenomenon is present despite the fact that higher $\mathcal{K}$ also allows emissions with larger numbers of molecules (see Table \ref{tab:paramtable}). Since ISI is considerably high in Figure \ref{fig:M_sweep_tb180}, further dividing the already short symbol duration into a large number of temporal sub-slots exacerbates the ISI issue between them. Thus, it can be inferred that the MCPM order $\mathcal{K}$ inherently governs a trade-off between the observed ISI and the transmission (hence also received) power.  Overall, given the considered system parameter set, $4$-MCPM is found to yield the lowest error probabilities among the evaluated schemes.

\par The results of Figure \ref{fig:M_sweep_tb300} suggest that MCPM still has an error performance improvement in the milder ISI regime, though the gain is less pronounced than the higher bit-rate scenario in Figure \ref{fig:M_sweep_tb180}. It can also be observed that as $M$ increases, pure PPM starts to become the more desirable choice than the proposed scheme, whilst in the lower transmission power range (small $M$), MCPM schemes outperform BCSK and PPM. Since MCPM schemes can encode more bits in a single symbol, they can emit more molecules for each symbol under the $M$ normalization, as also presented in Table \ref{tab:paramtable}. This property is especially beneficial when $M$ is smaller, as it helps MCPM schemes avoid extremely low emission intensities better than its competitors in this regime.

\begin{figure}[!t]
	\centering
	\includegraphics[width=0.48\textwidth]{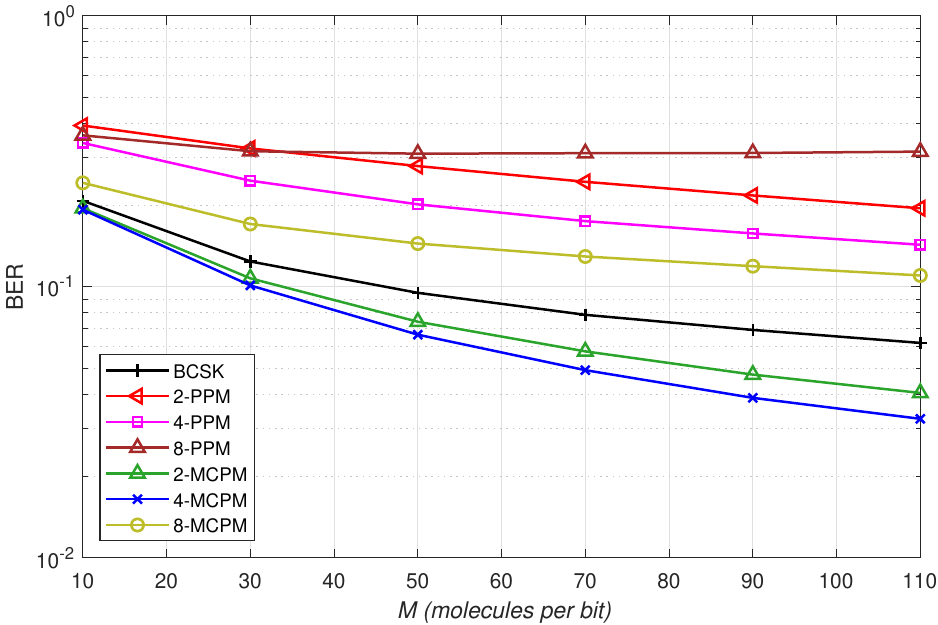} 
	\caption{BER vs. $M$ curves for MCPM and competing schemes. $t_b = \SI{0.18}{\second}$, $r_0 = \SI{10}{\micro\meter}$, $r_r = \SI{5}{\micro\meter}$, $D = \SI{79.4}{\micro\meter\squared\per\second}$, and $t_{\textrm{total}} = 48t_b$.}	
	\label{fig:M_sweep_tb180}
\end{figure}

\begin{figure}[!t]
	\centering
	\includegraphics[width=0.48\textwidth]{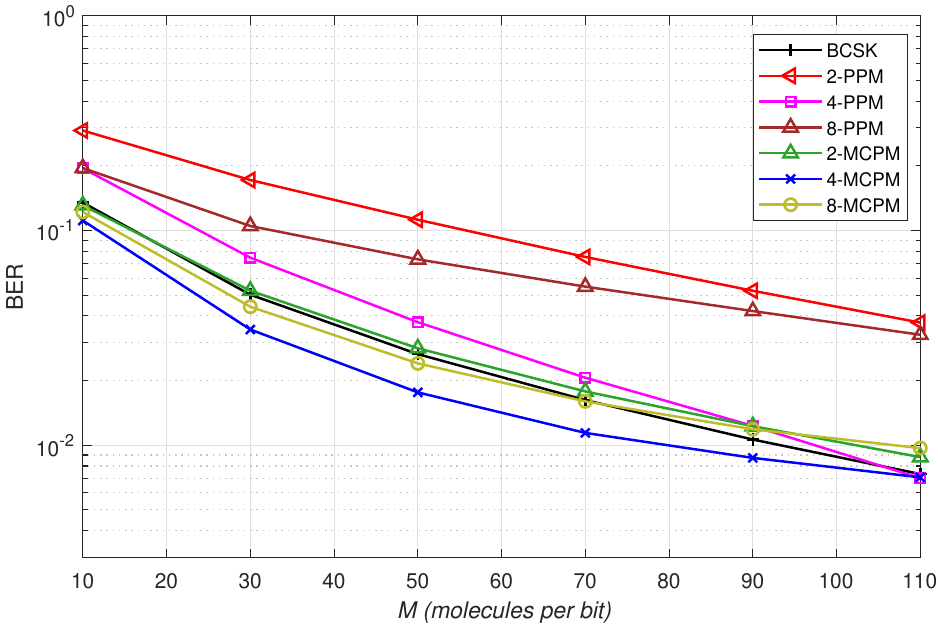} 
	\caption{BER vs. $M$ curves for MCPM and competing schemes. $t_b = \SI{0.30}{\second}$, $r_0 = \SI{10}{\micro\meter}$, $r_r = \SI{5}{\micro\meter}$, $D = \SI{79.4}{\micro\meter\squared\per\second}$, and $t_{\textrm{total}} = 48t_b$.}	
	\label{fig:M_sweep_tb300}
\end{figure}

\subsection{Error Performance and Bit-Rate}

\par The results and discussion presented in Subsection \ref{subsec:Msweep} suggest that MCPM compares to its concentration and position components differently under different amounts of ISI. Since the ISI in an MCD system depends on the relationship between the topological parameters and the symbol duration, Figure \ref{fig:tb_sweep} is presented to evaluate the error performance of MCPM with respect to $t_b$. 

\begin{figure}[!t]
	\centering
	\includegraphics[width=0.48\textwidth]{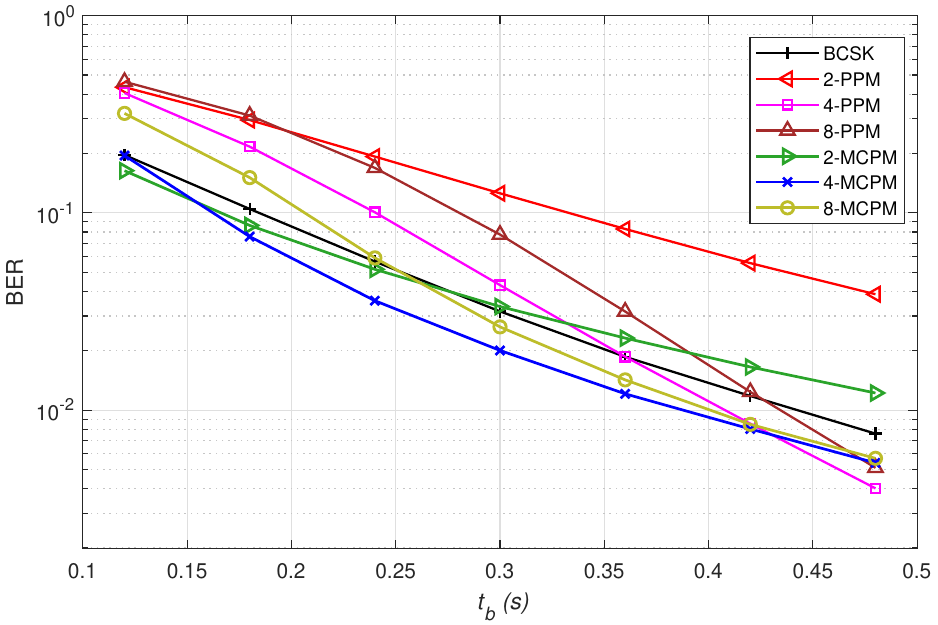} 
	\caption{BER vs. $t_b$ curves for MCPM and competing schemes. $M = 50$ molecules, $r_0 = \SI{10}{\micro\meter}$, $r_r = \SI{5}{\micro\meter}$, $D = \SI{79.4}{\micro\meter\squared\per\second}$, and $t_{\textrm{total}} = \SI{30}{\second}$.}	
	\label{fig:tb_sweep}
\end{figure}

\par The results of Figure \ref{fig:tb_sweep} show that even though pure PPM is a better strategy for larger $t_b$ scenarios, at least one order of MCPM outperforms the existing schemes in high bit-rate scenarios. Combined with the results presented in Subsection \ref{subsec:Msweep}, MCPM is found to be especially beneficial when the bit-rate is relatively high (small $t_b$) and the transmission power is relatively low (small $M$), which suggests possible applications for MCD under these channel conditions and with simple nano-scale machinery. Furthermore, noting that the results presented in Figures \ref{fig:M_theosim_tb300}-\ref{fig:tb_theosim_M50} suggest the proposed $\alpha$ and $\gamma$ selection techniques are accurate in the mentioned regime, it follows that within the parameter regimes where MCPM is beneficial, the presented numerical results for MCPM can be closely approximated using the proposed methods.

\subsection{Error Performance under Imperfect Synchronization}

\par Although perfect synchronization between the transmitter and the receiver was assumed until this point, this assumption may not always hold in an MCD link \cite{sync_error1}. In this section, the robustness and general behavior of MCPM under mis-synchronization are investigated. Specifically, the scenario in which the receiver's clock lags behind the transmitter's clock is examined, where the parameter $\tau$ denotes the clock offset between them. The $(\alpha,\gamma)$ pairs for the MCPM schemes, alongside the threshold values for BCSK are numerically optimized for each $\tau$. Using these definitions and considerations, Figure \ref{fig:tau_sweep} presents the obtained BER versus $\tau$ curves.

\par The results of Figure \ref{fig:tau_sweep} imply that in the very high $\tau$ regime (large synchronization error), the MCPM and PPM schemes with higher orders outperform their lower order counterparts. This phenomenon is mainly due to allowing a larger symbol duration while still satisfying the same $t_b$ constraint, since having a larger $t_{sym}$ makes $\tau$ smaller with respect to the symbol duration. This effect also explains the phenomenon of MCPM outperforming pure PPM with the same order, since the concentration dimension allows to encode more bits into a single symbol than pure PPM. 

\begin{figure}[!t]
	\centering
	\includegraphics[width=0.48\textwidth]{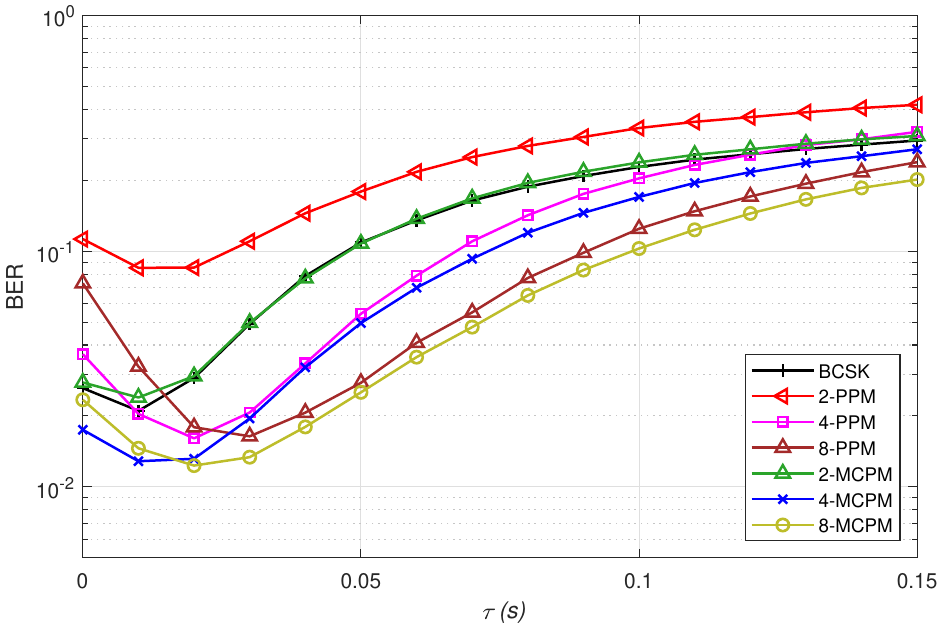} 
	\caption{BER vs. $\tau$ curves for MCPM and competing schemes. $M = 50$ molecules, $t_b = \SI{0.30}{\second}$, $r_0 = \SI{10}{\micro\meter}$, $r_r = \SI{5}{\micro\meter}$, $D = \SI{79.4}{\micro\meter\squared\per\second}$, $t_{\textrm{total}} = 48t_b$.}	
	\label{fig:tau_sweep}
\end{figure}

\par An interesting observation in Figure \ref{fig:tau_sweep} is that the error performance does monotonically deteriorate as $\tau$ increases. This phenomenon has been documented in the literature for BCSK, in the context of deliberately adding a reception delay to improve error performance \cite{reception_delay}. The results of Figure \ref{fig:tau_sweep} suggest that the BER improvement introduced by a reception delay is also applicable to timing-based modulation schemes like MCPM and PPM as well. The main reason for this beneficial phenomenon is linked to the behavior of the arrival time distribution. There is a non-negligible delay between the emission instant and the channel peak time \cite{3Dchar}, which implies that within this time interval, ISI contributes to the received signal more than the intended symbol does. Having a $\tau$ delay in the reception window can help the receiver to mitigate ISI (at the cost of reducing the received signal power) and be beneficial. Of course, further increasing $\tau$ causes losing the intended symbol's contribution as well, resulting in a poorer error performance.

\section{Conclusions}
\label{sec:conclusion}

\par In this paper, the problem of MCD modulation design using a single type of molecules has been addressed. To this end, a hybrid modulation scheme, \textit{molecular concentration position modulation} (MCPM), has been proposed. The MLSD for the proposed modulation has been derived. Motivated by the high computational complexity of MLSD, a reduced complexity sub-optimal detector (\textit{i.e.}, TPCD) has been introduced. For an MCPM scheme utilizing TPCD, the problem of constellation point design has been addressed through the parameter $\alpha$. In order to optimize $\alpha$, a cost function has been proposed and was shown to be convex in $\alpha$. Furthermore, a low-complexity sub-optimal strategy to select the threshold used in TPCD was presented. Overall, our numerical results suggest that MCPM schemes outperform BCSK and PPM especially in the high bit-rate and low transmission power regime, and provide better robustness against synchronization errors between the transmitter and the receiver.

\appendices
\section{Proof of Lemma 1}
\label{ap:gamma_cvx}

\par The proof simply follows from showing that $\frac{\partial^2 C}{\partial \gamma_{U}^2} > 0$ holds within the valid region. The second derivative is found as
\begin{equation}
    \begin{split}
    	\frac{\partial^2 C}{\partial \gamma_{U}^2} = \frac{1}{\sqrt{2 \pi}} \Bigg( &\frac{\gamma_{U} - B(1-\alpha_{U})h_1}{(B(1-\alpha_{U})h_1)^{3/2}} e^{-\frac{(\gamma_{U} - B(1-\alpha_{U})h_1)^2}{2B(1-\alpha_{U})h_1}} \\ & \hspace{1cm} - \frac{\gamma_{U} - B\alpha_{U}h_1}{(B \alpha_{U} h_1)^{3/2}} e^{-\frac{(\gamma_{U} - B \alpha_{U}h_1)^2}{2B \alpha_{U} h_1}}  \Bigg).        
    \end{split}
    \label{eq:gamma_cvx_prf}
\end{equation}
Note that the exponential expressions are always positive, $(B\alpha_{U}h_1)$ and $B (1-\alpha_{U})h_1$ are positive by definition, $\gamma_{U} - B(1-\alpha_{U})h_1$ is positive and $\gamma_{U} - B \alpha_{U} h_1$ is negative due to the defined valid region of $\gamma_{U}$ (first assumption in Equation \eqref{eq:assumptions}). Overall, \eqref{eq:gamma_cvx_prf} involves subtracting a negative quantity from a positive, thus we have $\frac{\partial^2 C}{\partial \gamma_{U}^2} > 0$. 

\section{Proof of Theorem 1}
\label{ap:alpha_cvx}

\par $C$ has four different ``types" of $Q$-function terms in it as summands, with a total of $2\mathcal{K}$ separate $Q$-functions. By plugging in $\gamma_{U} = \gamma^{*}_U(\alpha_{U})$ from \eqref{eq:optimal_gamma}, each of these $Q$-functions can be shown to be convex in $\alpha_{U}$.

\begin{itemize}
    \item Let $E_i = Q\bigg( \frac{B \alpha_{U} (h_1 - h_i)}{ \sqrt{B \alpha_{U} (h_1 +h_i)}} \bigg)$: The convexity of this part follows from the fact that their derivative is an increasing function of $\alpha_{U}$. Given the constraints in \eqref{eq:assumptions} are met,
		\begin{equation}
			\frac{\partial E_i}{\partial \alpha_{U}} = - \frac{B(h_1 - h_i)e^{-\frac{B\alpha_{U}(h_1-h_i)^2}{2(h_1+h_i)}}}{\sqrt{8 \pi B \alpha_U (h_1+h_i)}} = -K_1 \frac{e^{-k_1 \alpha_{U}}}{\sqrt{\alpha_{U}}}
		\end{equation}
		where $K_1$ and $k_1$ are positive coefficients. $e^{-k_1 \alpha_{U}}$ is a decreasing function of $\alpha_{U}$, whereas $\sqrt{\alpha_{U}}$ is an increasing function of it. The division of a decreasing and an increasing function is decreasing, and the minus sign makes the whole expression increasing.
		
		\item Let $F_i = Q\bigg( \frac{B (1-\alpha_{U}) (h_1 - h_i)}{ \sqrt{B (1-\alpha_{U}) (h_1 +h_i)}} \bigg)$: The convexity of this part is proven in a very similar manner to $E_i$, by showing that $\frac{\partial F_i}{\partial \alpha_{U}}$ is an increasing function of $\alpha_{U}$.
		
		\item Let $G = Q\bigg( \frac{B\alpha_{U} h_1 - \gamma_U}{ \sqrt{B \alpha_{U} h_1}} \bigg)$: The proof of this expression follows from its second derivative being always positive. The second derivative can be written as:
		\begin{equation}
		\label{eq:alphaconvex_Gi}
		\begin{split}
		    &\frac{\partial^2 G}{\partial \alpha_{U}^2} = \frac{e^{-\frac{(B\alpha_Uh_1 - \gamma_U)^2}{2B\alpha_Uh_1}}}{\sqrt{32\pi} \alpha_U^2 (B\alpha_Uh_1)^{\frac{3}{2}} } \times \\
		    & \Big( -\gamma_U^3 - (B\alpha_Uh_1)\gamma_U^2 + (B\alpha_U h_1)^3 \\
		    &+ (B\alpha_U h_1)^2 + (B\alpha_U h_1)^2\gamma_U + 3(B\alpha_U h_1) \Big).
		\end{split}
		\end{equation}
		The fraction in the first row of \eqref{eq:alphaconvex_Gi} is always positive, hence it is sufficient to show the positivity of the expression within the parentheses. Since the assumptions in \eqref{eq:assumptions} imply $\gamma_U < B\alpha_U h_1$, the inequalities $(B\alpha_U h_1)^3>(B\alpha_Uh_1)\gamma_U^2$ and $(B\alpha_U h_1)^2\gamma_U > \gamma_U^3$ hold and said expression is also guaranteed to be positive, making $\frac{\partial^2 G}{\partial \alpha_{U}^2} > 0$.

		\item Let $H = Q\bigg( \frac{B(1-\alpha_{U}) h_1 - \gamma_{U}}{ \sqrt{B (1-\alpha) h_1}} \bigg)$: Similar to $G$, the proof of this expression follows from $\frac{\partial^2 H}{\partial \alpha_{U}^2} > 0$. Unfortunately, unlike the earlier $Q$-function types, $H$ is not always convex in $\alpha_U$ for all valid $\gamma_U$. However, given the conditions in \eqref{eq:assumptions} are met, $H$ is indeed convex in $\alpha_U$ when $\gamma_U = \gamma^{*}_U(\alpha_{U})$. The second derivative evaluated at $\gamma^{*}_U$ can be written as
		\begin{equation}
		\label{eq:alphaconvex_Hi}
		\begin{split}
			\frac{\partial^2 H}{\partial \alpha_{U}^2} \Bigg|_{\gamma_{U} = \gamma^{*}_{U}} \hspace{-0.5cm} &= \frac{e^{- \frac{(\gamma^{*}_{U} - B\alpha_U h_1)^2}{2B(1-\alpha_U)h_1}}}{(1-\alpha_{U})^2 (2 \alpha_{U}-1) \sqrt{32\pi B(1-\alpha_{U}) h_1}} \times \\
			 \Bigg[ (1-&2\alpha_{U})  \Big( B^2 h_1^2(2\alpha_{U}^2-3\alpha_U+1) + \\
			B&(1-\alpha_{U})h_1 + Bh_1 (1-2\alpha_{U})\gamma^{*}_{U} +3\gamma^{*}_{U} \Big) \\ 
			&- \alpha_{U} \Big( B(1-\alpha_{U})h_1 + \gamma^{*}_{U} \Big) \ln\big(\frac{1-\alpha_{U}}{\alpha_{U}}\big) \Bigg].
		\end{split}
		\end{equation}
        Similar to $G$, the fraction in the first row of \eqref{eq:alphaconvex_Hi} is always positive. Thus, showing that the expression in the square brackets is positive is sufficient to ensure the convexity of $H$. Rearranging the terms, the positivity of the said part is guaranteed when the following inequality holds:
		\begin{equation}
    		\gamma^{*}_{U} > B(1-\alpha_{U})h_1 \frac{1 - Bh_1 (2\alpha_{U} -1)}{Bh_1 (2\alpha_{U} -1) - 3}.
    		\label{eq:araform}
    	\end{equation}
		By definition, $\gamma^{*}_{U} > B(1-\alpha_{U})h_1$. Thus, ensuring $\frac{1 - Bh_1 (2\alpha_{U} -1)}{Bh_1 (2\alpha_{U} -1) - 3}  \leq 1$ is sufficient for \eqref{eq:araform} to hold. This implies 
			\begin{itemize}
			    \item $B \alpha_{U} h_1 - B (1-\alpha_{U}) h_1 > 3$, or
			    \item $B \alpha_{U} h_1 - B (1-\alpha_{U}) h_1 \leq 2$,
			\end{itemize}
		where the former is within the assumption set in \eqref{eq:assumptions}. 
\end{itemize}
Lastly, since each summand in $C$ is convex given the assumption set \eqref{eq:assumptions} is satisfied, the sum of these convex functions (\textit{i.e.}, $C$) is also convex.

\bibliographystyle{IEEEtran}
\bibliography{refs_new}

\end{document}